\documentclass[prl,letterpaper,twocolumn]{revtex4}
\usepackage[latin1]{inputenc}
\usepackage{amsmath,amsbsy}
\usepackage{amsfonts,hyperref}
\usepackage{bbm}
\usepackage{verbatim}
\usepackage{amsthm}

\newcommand{\ket}[1]{|#1\rangle}
\newcommand{\bra}[1]{\langle #1|}
\newcommand{\bracket}[2]{\langle #1|#2\rangle}
\newcommand{\ketbra}[1]{|#1\rangle\langle #1|}
\newcommand{\cket}[1]{|\widetilde{#1}\rangle}

\newtheorem*{theorem}{Decoupling Theorem}

\usepackage{amsmath}
\usepackage{amssymb}

\begin{document}
\author{Joseph M.~Renes$^1$ and Jean-Christian Boileau$^{2}$}
\affiliation{$^1$Institut f\"ur Angewandte Physik, Technische Universit\"at Darmstadt, Hochschulstr.~4a, 64289 Darmstadt, Germany\\
$^2$Center for Quantum Information and Quantum Control, University of Toronto, Toronto, ON, M5S 1A7 Canada}

\title{Conjectured Strong Complementary Information Tradeoff}

\begin{abstract}
We conjecture a new entropic uncertainty principle governing the entropy of complementary observations made on a system given side information in the form of quantum states, 
generalizing the entropic uncertainty relation of Maassen and Uffink [Phys.\ Rev.\ Lett.\ {\bf 60}, 1103 (1988)]. We prove a special case for certain 
conjugate observables 
by adapting a similar result found by Christandl and Winter pertaining to quantum channels [IEEE Trans.\ Inf.\ Theory {\bf 51}, 3159 (2005)], and
discuss possible applications of this result 
to the decoupling of quantum systems and for security analysis in quantum cryptography. 
\end{abstract}

\maketitle 

One of the central mysteries of quantum mechanics is complementarity, the strange 
phenomenon that a given physical attribute can only be exhibited at the expense of 
another, complementary, attribute. The canonical example, wave-particle duality, 
is illustrated in the double slit experiment. Coherent light (or matter) travelling 
through both slits produces an interference pattern, a wave-like property which is 
however destroyed if one determines which path has been taken, a particle-like property. 
Such behavior vividly differentiates quantum from classical mechanics and led Feynman to 
famously observe that the double slit experiment ``has in it the heart of 
quantum mechanics
''~\cite{feynman_feynman_1970}.

The new field of quantum information theory takes a pragmatic approach to the 
mysteries of quantum mechanics, seeking to better understand them by asking which 
information processing tasks can or cannot be accomplished in this new arena. The 
results have been stunning. Quantum information cannot be copied but can be ``teleported'' 
from place to place. It can be used to improve the precision of 
everything from clock synchronization to gravitational-wave detectors to lithography. 
It can dramatically speed up certain computational tasks, 
such as searching an unordered list and factoring large integers. 
Quantum information cannot, however, be shared between many parties.
For instance, maximal entanglement can be shared by only two parties, and entangling 
more parties means making the entanglement between any two of them weaker. 
This effect also enables cryptographic tasks which are impossible classically, such as unconditionally 
secure key exchange. This property of 
exclusiveness or privacy informs many 
aspects of how we reason about quantum information and quantum information 
processing~\cite{nielsen_quantum_2000}. 

The connection between complementarity and privacy stems from the entropic uncertainty relation due
to Maassen and Uffink~\cite{maassen_generalized_1988}, then successively extended by Hall~\cite{hall_information_1995} and Cerf \emph{et al}.~\cite{cerf_security_2002}. The original version
constrains the entropies of two noncommuting observables $O^A$ and $\widetilde{O}^A$ of a system $A$, and the latter versions extend
this to explicitly include classical side information about the observables, stored either jointly in one external system $R$ (Hall) or separately in two, $B$ and $E$ (Cerf \emph{et al.}). These external systems might, for example, be ancillary systems used in von Neumann measurement processes, possibly of $O^A$ or $\widetilde{O}^A$ or both.
Giving the $B$ and $E$ systems to parties Bob and Eve, respectively (the names are chosen in anticipation of the cryptographic implications to follow), the complementarity statement of Cerf \emph{et al}.~says that the information one party (Bob) could obtain about one observable ($O^A$) by measuring his system $B$, plus the information Eve could obtain about the other observable ($\widetilde{O}^A$) by measuring $E$, cannot exceed  a prescribed bound. 
Equivalently, one can say that there is a certain unavoidable amount of uncertainty or entropy about the two observables conditioned on respective measurements of the two systems $B$ and $E$. 

In this paper we generalize the tradeoff to restrict the amount of conditional entropy 
the parties can have about noncommuting observables on $A$ when they possess \emph{quantum} side information. Quantum and classical side information behave differently, and in particular the information represented by the quantum state may be significantly larger than the amount of classical information that can be extracted from it by measurement, a statement known as the Holevo bound~\cite{holevo_bounds_1973,nielsen_quantum_2000}. Relatedly, classical side information is subject to \emph{locking}, meaning that a modest amount of additional classical side information can greatly increase the total~\cite{divincenzo_locking_2004}. Quantum side information, in contrast, cannot be locked in this manner.
We find numerical evidence for the generalized tradeoff for arbitrary observables and provide a proof for conjugate observables~\footnote{Two observables are {\it conjugate} if each of the eigenvectors of one is an equal-weight superposition of all the eigenvectors of the other observable.} related by a Fourier transform  based on the proof of a related 
entropic inequality for quantum channels given by Christandl and 
Winter~\cite{christandl_uncertainty_2005}. 

We then exhibit a family of states which saturate the bound before discussing some applications of our result. Building on~\cite{christandl_uncertainty_2005, hayden_random_2008}, we derive a rigorous statement of the idea that if the $AB$ system has nearly maximal quantum correlations, as 
measured by appropriately small quantum conditional entropies,  then the $AE$ system is nearly decoupled, i.e.~in a product state. This is akin to the monogamy of entanglement, the
fact that maximal entanglement cannot be shared by more than two parties, at the level of quantum correlations.  
In the context of quantum cryptography, 
this means composable security---that the  key generated in quantum key distribution is secure in any further cryptographic application~\cite{ben-or_universal_2005,renner_universally_2005}---can be established 
by bounding the 
information that the (quantum) system $B$ has on a basis conjugate to the basis of $A$ used to 
encode the key.

{\em Basic entropic uncertainty principles.}---We begin by reviewing the existing entropic uncertainty principles. For a physical system $A$, consider any two observables $O^A$ and $\widetilde{O}^A$
represented by operators on a finite-dimensional Hilbert space. Let $c=\max_{jk}|\bracket{j}{\widetilde{k}}|$, where $\ket{j}$ and $\cket{k}$ are the eigenvectors of $O^A$ and $\widetilde{O}^A$, respectively. Define $H(O^A)_\rho$ and $H(\widetilde{O}^A)_\rho$ to be the respective Shannon entropies
of the outcome probabilities of the measurements of $O^A$ and $\widetilde{O}^A$ on a given state $\rho^A$. Maassen and Uffink~\cite{maassen_generalized_1988} showed that regardless of $\rho^A$,
\begin{align}
\label{eq:eur}
H(O^A)_\rho+H(\widetilde{O}^A)_\rho \geq -2\log_2 c,
\end{align}
meaning the entropies of  
observables sharing no common eigenstates cannot both be arbitrarily small. The constant $c$ 
can be as large as $\log_2 d^A$, for $d^A$ the dimension of the Hilbert space describing system 
$A$, and is exactly $\log_2 d^A$ if and only if the observables are \emph{conjugate}. 
For conjugate observables, certainty regarding one observable 
implies complete uncertainty regarding the other.

Suppose now that we have some side information or background information relevant to $A$, for instance 
the result of measurement on some external system $R$ which is correlated with $A$. Intuitively,
the entropic uncertainty relation should still hold, since this information would simply 
factor into the description $\rho^A$ of $A$. Indeed, the entropic uncertainty principle 
can be adapted to this case, and is equivalent to a result by Hall which he terms
the information-exclusion principle~\cite{hall_information_1995}. Its derivation proceeds as follows. 

Consider an arbitrary bipartite 
quantum state $\rho^{AR}$ where $A$ and $R$ are two finite dimensional Hilbert spaces. Let 
$\Gamma^R$ be a positive operator-valued measure (POVM) representing an arbitrary measurement on the system $R$ (i.e. 
$\sum_j \Gamma^R_j=\mathbbm{1}^R$ and $\Gamma^R_j \geq 0$ for all $j$). 
Measurement of $\Gamma^R$ gives the outcome $j$ with probability 
$q_j={\rm Tr}[\rho^{AR} \Gamma^R_j]$, and leaves the marginal state of $A$ given by 
$\rho^A_j:= \frac{1}{q_j}{\rm Tr}_R[\rho^{AR} \Gamma^R_j]$~\footnote{Note that Hall 
simply starts with the ensemble $\mathcal{E}=\{q_j, \rho^A_j\}$, whereas here it arises from
measurement of $R$.}. 
Applying the inequality Eq.~\ref{eq:eur} 
to each of the states $\rho_j$ gives $H (O^A)_{\rho_j}+H (\widetilde{O}^A)_{\rho_j} \geq -2\log_2 c$, where
the entropies are computed using the conditional state $\rho_j$. Since this state is determined by the
classical outcome of measurement on $R$, we can write $H(O^A)_{\rho_j}$ as $H(O^A|\Gamma_j^R)_\rho$ and likewise for observable $\widetilde{O}^A$, where $H(O^A|\Gamma_j^R)_\rho$ is the conditional entropy of the $O^A$ observable given the result of the $\Gamma^R$ measurement. Averaging over all outcomes yields the information exclusion principle:
\begin{align}
\label{eq:iep}
H (O^A | \Gamma^R)_\rho+H (\widetilde{O}^A|\Gamma^R)_\rho \geq -2\log_2 c.
\end{align}

This equation is particularly useful if we consider $R$ to be a composite system, consisting of subsystems $B$ and $E$, and $\Gamma^R$ to be a composite measurement $\Gamma^R_{jk}=\Lambda^B_j\otimes \widetilde{\Lambda}^E_k$, as put forth in~\cite{cerf_security_2002}. Since conditioning reduces entropy, one obtains a tradeoff in the amount of information about $O^A$ and $\widetilde{O}^A$ which can be simultaneously stored in \emph{separate} auxiliary systems $B$ and $E$. We call this the
(weak) complementary information tradeoff (CIT):
\begin{align}
\label{eq:iep2}
H (O^A | \Lambda^B)_\rho+H (\widetilde{O}^A|\widetilde{\Lambda}^E)_\rho \geq -2\log_2 c.
\end{align}
Now the information held by one party, in possession of system $B$, say, limits the information which another party could in
principle obtain about a noncommuting observable. This tradeoff is immediately applicable in quantum cryptography, and 
 in~\cite{renes_physical_2008} we used it to 
motivate a new approach to the distillation of entanglement and secret keys.
Our present goal is to find a stricter tradeoff.

{\it Strong Complementary Information Tradeoff}.---What if we regard the quantum state of the auxiliary system itself as the side information? Is there any limit to the uncertainty of complementary observables in this case? 
One might conjecture that the quantum version of Eq.\ (\ref{eq:iep}) holds, replacing the conditional Shannon 
entropy $H(O^A|\Gamma^R)$ with the conditional von Neumann 
entropy $S(O^A|R)_\rho=S(\rho^{AR}_{O^A})-S(\rho^R)$,
where $\rho^{AR}_{O^A}$ is the quantum state obtained after measuring the observable $O$ on the state $\rho$ and averaging over all outcomes. However, this is false in general.
To take an extreme example, the singlet state of two spin-$\frac{1}{2}$ particles is
perfectly anticorrelated in every basis, meaning that $S(O^A|R)=S(\widetilde{O}^A|R)=0$ for \emph{any}  nondegenerate observables. This is merely the statement that quantum correlations, in the form of 
entanglement, are stronger than classical correlations. 
 
Instead, the conjecture should be applied to the weak CIT, Eq.~(\ref{eq:iep2}), and the result is the strong CIT:
\begin{align}
\label{eq:cit}
S(O^A|B)_\rho+S(\widetilde{O}^A|E)_\rho\geq -2\log_2 c.
\end{align}
The strong CIT immediately implies the weak CIT via the Holevo bound $S(O^A|\Gamma^B)\geq S(O^A|B)$ for any measurement $\Gamma$~\cite{holevo_bounds_1973,nielsen_quantum_2000}, and also the original entropic uncertainty principle by taking $B$ and $E$ to be one-dimensional.

The claim is supported by 
numerical investigation of small dimensions $d_A,d_B,d_E\leq 12$, which has  
found no counterexample when testing at least 2000 random states in each of the $11^3$ combinations of 
dimensions. By itself this is relatively weak evidence, but for conjugate observables 
$O^A, \widetilde{O}^A$ related by a Fourier transform, e.g.~$\widetilde{O}^A=F^AO^AF^{A\dagger}$, Eq.~(\ref{eq:cit}) follows from strong subadditivity (SSA) of the von Neumann entropy. Assuming that the eigenvectors of 
$O^A$ define a standard basis, we can redefine 
the eigenvalues of the observables so that $O\rightarrow Z=\sum_k \omega^{k}\ketbra{k}$ and $\widetilde{O}\rightarrow X=\sum_k \ket{k{+}1}\bra{k}$, where $\omega=e^{2\pi i/d}$, the generalized Pauli operators~\footnote{These operators
are not Hermitian, but as only the eigenvectors concern us here, this is of no consequence.}.
Then the various properties of $X$ and $Z$ can be used to construct a proof, in a manner entirely 
similar to~\cite{christandl_uncertainty_2005}, who establish a similar tradeoff for the ability of a quantum channel to transmit conjugate information. 
In the present context, the proof goes as follows.

\begin{proof}[Proof of special case]
Consider any $\rho^{ABE}$ where ${\rm dim}(A)=d$. We can assume that $\rho^{ABE}$ is a pure state without loss of generality, since $E$ can always be redefined to include the purification. 
As a consequence of SSA, the value of $S(O^A| E)_\rho$ cannot be 
increased by enlarging $E$ (cf.~\cite{nielsen_quantum_2000}, Theorem 11.15). 
So if the inequality is true for any pure state then 
it must also be true for any mixed state. Using the properties of $X$ and $Z$ one can write 
$\bar{\rho}^{ABE}_{X^A} = \frac{1}{d}\sum_kX_k^A\rho^{ABE}X_k^{\dagger A}$ and 
$\bar{\rho}^{ABE}_{Z^A} = \frac{1}{d}\sum_kZ_k^A\rho^{ABE}Z_k^{\dagger A}$. Here we have used a 
nonstandard notation, defining $X^A_k:=(X^A)^k$ and $Z^A_k:=(Z^A)^k$.
Now let $\rho_{jk}^{ABE}:=X_j^AZ_k^A\rho^{ABE}Z_k^{\dagger A}X_j^{\dagger A}$ and 
define 
\begin{align}
\Omega^{A'\!B'\!AB}:=\frac{1}{d^2}\sum_{jk} P_j^{A'}\!\otimes P_k^{B'}\!\otimes \rho^{AB}_{jk},
\end{align}
for $P_j=\ketbra{j}$. $A'$ and $B'$ are two new systems such that ${\rm dim}(A')={\rm dim}(B')=d$. The sum of $j$ and $k$ is understood to be over all values from 1 to $d$.  Direct calculation shows that 
\begin{align}
S(A'|AB)_\Omega &= S(\bar{\rho}_{Z^A}^{AB})-S(\rho^B)=S(Z^A|B)_\rho\\
S(B'|AB)_\Omega &= S(\bar{\rho}_{X^A}^{AB})-S(\rho^B)=S(X^A|B)_\rho\\
S(A'B'|AB)_\Omega&=\log_2 d+S(A|B)_\rho.
\end{align}

Strong subadditivity is just the statement that $S(A'B'|AB)\leq S(A'|AB)+S(B'|AB)$ (cf.~\cite{nielsen_quantum_2000}, Theorem 11.16), so
\begin{align}
S(Z^A|B)_ \rho+S(X^A|B)_ \rho \geq \log_2 d+S(A|B)_ \rho.
\end{align}

Define the probability distribution $p_k$ and quantum states $\ket{\varphi_k}^{BE}$ such that
$\ket{\rho}^{ABE}= \sum_{k} \sqrt{p_k} \ket{k}^A \ket{\varphi_k}^{BE}$. Using $S(\bar{\rho}^{AB}_{Z^A})=H(p_k)+\sum_k p_k S(\varphi_k^B)=H(p_k)+\sum_k p_k S(\varphi_k^E)=S(\bar{\rho}^{AE}_{Z^A})$, a simple calculation reveals that $S(Z^A|B)_\rho-S(Z^A|E)_\rho=S(A|B)_\rho$ and hence
for an arbitrary pure $\rho^{ABE}$,
\begin{align*}
S(X^A|B)_\rho +S(Z^A|E)_\rho &\geq \log_2 d.\qedhere
\end{align*}
\end{proof} 

{\it Saturating the Bound.}---
Since the bound relies solely on SSA, saturating the bound means fulfilling the SSA
equality conditions. A useful form of these is given in~\cite{hayden_structure_2004}, which states in the present case that
the $AB$ state space must decompose as $\mathcal{H}^{AB}\simeq \bigoplus_{s}\mathcal{H}^{L_s}\otimes \mathcal{H}^{R_s}$, so that 
\begin{align}
\Omega^{A'B'AB}=\bigoplus_s r_s \sigma^{A'L_s}\otimes\omega^{B'R_s}
\end{align}
for some states $\sigma^{A'L_s}$, $\omega^{B'R_s}$ and probabilities $r_s$. Projecting $A'B'$ onto the $jk$th outcome gives  
$\rho_{jk}=d^2\bigoplus_s r_s \sigma_j^{L_s}\otimes \omega_k^{R_s}$, where $\sigma_j^{L_s}={\rm Tr}[P_j^{A'}\sigma^{A'L_s}]$ and similarly for $\omega_k^{R_s}$.
Thus, the action of $X_j^A$ and $Z_k^A$ on $\rho^{AB}$ must be on different subsystems within each $s$-sector. 

One way to arrange for this is to take $\rho^{AB}$ to be one of the Bell-states $\ket{\Phi_{jk}}^{AB}=\frac{1}{\sqrt{2}}X_j^BZ_k^B(\ket{00}+\ket{11})^{AB}$. Then there is only one $s$-sector, and the spaces $\mathcal{H}^{L}$, $\mathcal{H}^R$ are two-dimensional, so that $\ket{\Phi_{jk}}\simeq \ket{jk}$.
Bell states saturate the bound in the most trivial manner possible: both $S(X^A|B)$ and $S(Z^A|B)=0$. 

A more interesting example is afforded by the state 
$\ket{\psi}^{ABE}=\frac{1}{\sqrt{d}}\sum_k\ket{k}^A\ket{\varphi_k}^{BE}$, where we set   
$\ket{\varphi_k}^{BE}=\sum_{uv}\sqrt{q_{uv}}\ket{u}^B\ket{\eta_{uv}}^{E_1}Z_k^{E_2}\ket{v}^{E_2}$ 
with arbitrary states $\ket{\eta_{u,v}}$ and distribution $q_{uv}$. 
Here the entropies $S(X^A|B)$ and $S(Z^A|E)$ do not necessarily take on 
extremal values, but their counterparts $S(Z^A|B)$ and $S(X^A|E)$ do. 
For starters, $S(Z^A|B)=\log_2 d$ since $\varphi_k^B$ is independent of $k$. Meanwhile, 
$S(X^A|E)=0$ can be quickly derived by making the substitution 
$\ket{k}^AZ_k^{E_2}\ket{v}^{E_2}=Z_v^A\ket{k}^A\ket{v}^{E_2}$ in the 
definition of $\ket{\psi}$. Thus, the $\rho^{AB}_{jk}$ derived from $\ket{\psi}$ 
meet the equality conditions, and therefore $S(X^A|B)+S(Z^A|E)= \log_2 d$. 

In more concrete terms, the $\rho^{AB}_{jk}$ meet the equality conditions because 
$\psi^{AB}=\sum_v q_v \widetilde{P}_v^A\otimes \xi_v^B$, with $\xi_v^B=\sum_{uu'}\sqrt{q_{u|v}q_{u'|v}}\bracket{\eta_{uv}}{\eta_{u'v}}\ket{u}\bra{u'}^B$. Thus, the action
of $X_j^A$ has no effect on $\psi^{AB}$, as it is already diagonal in the $X^A$ basis.
Therefore we need only keep one $s$-sector and can  
dispense entirely with $\mathcal{H}^{L_s}$ in the decomposition of $\rho^{AB}_{jk}$, 
setting $\mathcal{H}^{R_s}=\mathcal{H}^{AB}$. It remains an open question whether any state 
can saturate the bound $S(X^A|B)+S(Z^A|E)\geq \log_2 d$ and its counterpart 
$S(Z^A|B)+S(X^A|E)\geq \log_2 d$ without taking on extremal values in either case.  

{\it Privacy Criterion.}---An immediate application of the strong CIT is in bounding the correlations between two systems $A$ and $E$, possessed by Alice and Eve, respectively, using the known correlations between $A$ and $B$, possessed by Bob. The weak form can also be used for this purpose, but the types of correlations that can be bounded are weaker as we now explain.

Consider the state $\rho^{ABE}$ and suppose that there exists a measurement $\widetilde{\Lambda}^B$ such that $H(\widetilde{O}^A|\widetilde{\Lambda}^B)_\rho\leq \epsilon$. This implies $H(O^A|{\Lambda}^E)_\rho\geq -2\log_2c-\epsilon$, or equivalently, $I(O^A{:}\Lambda^E)\leq \epsilon+H(O^A)_\rho+2\log_2c$. Supposing further that
$H(O^A)_\rho\leq -2\log_2 c$, as would necessarily be the case for conjugate observables, we obtain a bound
on Eve's information about $O^A$: $I(O^A{:}\Lambda^E)\leq \epsilon$. Applied to a quantum key distribution scenario where Alice's key is given by the measurement of the observable $O^A$, this
ensures a certain level of privacy of the key~\cite{cerf_security_2002}. However, due to locking, this security criterion 
is not {\it universally composable}~\cite{koenig_small_2007}, meaning that the key cannot be safely used in arbitrary further cryptographic protocols. For a precise definition of universal composability and an exhaustive explanation on why a suitable security criterion should be composable, see~\cite{ben-or_universal_2005,renner_universally_2005}.

On the other hand, the strong CIT \emph{can} be used to obtain a composable security criterion. The same conditions as above now imply that the eavesdropper's Holevo information is small, $I(O^A{:}E)\leq \epsilon$, which is a composable security criterion~\cite{ben-or_universal_2005}.
We can use the strong CIT to give an even more direct statement, in the form of sufficient conditions for decoupling Alice from Eve.
\begin{theorem}
Suppose $\rho^{ABE}$ is a tripartite state subject to the conditions $S(\rho^A)\leq -2\log_2 c$, $S(O^A|  B)_\rho \leq \epsilon_1$, and $S(\widetilde{O}^A|  B)_\rho \leq \epsilon_2$. Then  
 \begin{align}
 {\rm Tr}\left|\rho^{AE}-\rho^A \otimes \rho^E\right| \leq 2 \sqrt{\epsilon_1 + \epsilon_2}.
 \end{align}
\end{theorem}
\begin{proof}
The proof assumes that the strong CIT holds for the observables $O^A$ and $\widetilde{O}^A$. Observe that the mixed state case follows from the pure state case since the trace distance cannot increase when removing the purifying system. Thus, we can assume that $\rho^{ABE}$ is a pure state. Then a straightforward calculation reveals that $S(\widetilde{O}^A|E)_\rho = S(A|E)_\rho+S(\widetilde{O}^A|B)_\rho$, using the fact that given the value of $\widetilde{O}^A$, the entropy of Eve's state is identical to the entropy of Bob's state. 
Using this to substitute for $S(\widetilde{O}^A|E)_\rho$ in the strong complementary information tradeoff, the three given conditions yield
$S(\rho^A)- S(A|E)_\rho \leq \epsilon_1+\epsilon_2$. This can be written as $S(\rho^{AE}|| \rho^A \otimes \rho^E) \leq \epsilon_1+\epsilon_2$, and since the relative entropy and the trace distance are related by $({\rm Tr}|\rho^{AE}- \rho^A \otimes \rho^E|)^2 \leq 4S(\rho^{AE}|| \rho^A \otimes \rho^E)$~\cite{ohya_quantum_1993}, this concludes the proof. 
\end{proof}
This theorem makes rigorous the intuition that full quantum correlations, in 
the sense of small quantum conditional entropy, between two systems $A$ and $B$ is equivalent to
being decoupled from any other system $E$. The same intuition has different, though related, rigorous expressions. When thinking of quantum correlations as entanglement, this goes under the heading 
of {monogamy of entanglement}~\cite{koashi_monogamy_2004}.   
Or, instead of using quantum mutual information, one can imagine
there exist measurements on $B$ which could predict both the outcome of $X^A$ and $Z^A$, and
the same sort of decoupling result holds~\cite{koashi_complementarity_2007,renes_physical_2008}.

{\it Conclusion.}---We have proposed a tradeoff in the amount of information simultaneously available about complementary observables, formulated in terms of the quantum conditional entropy. It can be seen as the
natural extension of the reformulation by Cerf \emph{et al.} of the information exclusion principle and the entropic uncertainty principle. It is also the ``static'' version, applicable to quantum states, of Christandl and Winter's ``dynamic'' conjugate information tradeoff, which is formulated for quantum channels. The proof of the latter leads immediately to a proof of the strong complementary information tradeoff, and numerical investigation reveals that the tradeoff appears to hold for arbitrary observables. 
We have also discussed conditions under which the tradeoff can be saturated, as well as described some applications to quantum cryptography and derived a decoupling criterion for quantum states.

It would be interesting to determine if a similar bound holds for the smoothed 
conditional min- and/or max-entropies, which are generalizations of the classical Renyi entropes of order $\infty$ and 1/2, respectively, and have direct operational interpretations~\cite{Koenig_operational_2008}. They are often relevant in studying information processing protocols at the single-shot rather than asymptotic level, and are therefore more fundamental. Note that the original derivation of Maassen and Uffink already gives the unconditional result $H_{\rm min}+H_{\rm max}\geq -2\log_2 c$.

{\it Acknowledgment.}---The authors acknowledge fruitful discussions with Gernot Alber, Matthias Christandl, Robert K\"onig, and Stephanie Wehner. JCB received support from NSERC and Quantumworks.

\bibliography{refs}
\bibliographystyle{apsrev}

\end{document}